\documentclass[12pt,reqno]{article}

\usepackage[usenames]{color}
\usepackage{amssymb}
\usepackage{graphicx}
\usepackage{amscd}

\usepackage[colorlinks=true,
linkcolor=webgreen,
filecolor=webbrown,
citecolor=webgreen]{hyperref}

\definecolor{webgreen}{rgb}{0,.5,0}
\definecolor{webbrown}{rgb}{.6,0,0}

\usepackage{color}
\usepackage{fullpage}
\usepackage{float}

\usepackage{graphics,amsmath,amssymb}
\usepackage{amsthm}
\usepackage{amsfonts}
\usepackage{latexsym}
\usepackage{epsf}

\usepackage{fullpage}

\newcommand{\seqnum}[1]{\href{http://oeis.org/#1}{\underline{#1}}}

\DeclareMathOperator{\od}{od}
\DeclareMathOperator{\ev}{ev}

\def\Enn{{\mathbb{N}}}
\def\Zee{{\mathbb{Z}}}
\def \nodiv{{|\kern-4.2pt/}}

\def\modd#1 #2{#1\ ({\rm mod}\ #2)}

\begin{document}

\theoremstyle{plain}
\newtheorem{theorem}{Theorem}
\newtheorem{corollary}[theorem]{Corollary}
\newtheorem{lemma}[theorem]{Lemma}
\newtheorem{proposition}[theorem]{Proposition}

\theoremstyle{definition}
\newtheorem{definition}[theorem]{Definition}
\newtheorem{example}[theorem]{Example}
\newtheorem{conjecture}[theorem]{Conjecture}

\theoremstyle{remark}
\newtheorem{remark}[theorem]{Remark}

\title{Discriminators and $k$-Regular Sequences}

\author{Sajed Haque and Jeffrey Shallit\\
School of Computer Science \\
University of Waterloo \\
Waterloo, ON  N2L 3G1 \\
Canada \\
{\tt s24haque@cs.uwaterloo.ca} \\
{\tt shallit@cs.uwaterloo.ca} 
}

\maketitle

\begin{abstract}
The discriminator of an integer sequence ${\bf s} = (s(i))_{i \geq 0}$,
introduced by Arnold, Benkoski, and McCabe in 1985,
is the map $D_{\bf s} (n)$ that sends $n \geq 1$ to the least positive
integer $m$ such that
the $n$ numbers $s(0), s(1), \ldots, s(n-1)$ are pairwise incongruent
modulo $m$.  In this note we consider the discriminators of a certain
class of sequences, the $k$-regular sequences.  We compute the
discriminators of two such sequences,
the so-called ``evil'' and ``odious'' numbers,
and show they are $2$-regular.  We give an example of a 
$k$-regular sequence whose discriminator is not $k$-regular.

Finally, we examine sequences that are their own discriminators,
and count the number of length-$n$ finite sequences with this property.
\end{abstract}

\section{Discriminators}
\label{intro}

Let ${\bf s} = (s(i))_{i \geq 0}$ be a sequence of distinct integers.
For each $n \geq 1$, if 
the $n$ numbers $s(0), s(1), \ldots, s(n-1)$ are pairwise
incongruent modulo $m$, we say that $m$ {\it discriminates} them.
For $n \geq 1$ we define 
$D_{\bf s} (n)$ to be the least positive integer $m$ that
discriminates
the numbers $s(0), s(1), \ldots, s(n-1)$; such an $m$ always exists
because of the distinctness requirement.  Furthermore, 
we set $D_{\bf s} (0) = 0$, but usually this will be of no consequence.
The function (or sequence) $D_{\bf s} (n)$ is
called the {\it discriminator} of the sequence $\bf s$, and was
introduced by Arnold, Benkoski, and McCabe \cite{ABM}.
They proved that the discriminator $D_{\rm sq}(n)$ of the sequence
$(n+1)^2_{n \geq 0} = 1, 4, 9, 16, \ldots$ of positive
integer squares is given by
$$ D_{\rm sq}(n) = \begin{cases}
	1, & \text{if $n = 1$}; \\
	2, & \text{if $n = 2$}; \\
	6, & \text{if $n = 3$}; \\
	9, & \text{if $n = 4$}; \\
	\min \{k \ : \ k \geq 2n \text{ and ($k = p$ or $k = 2p$ for
		some prime $p$) } \}, & \text{if $n > 4$}.
	\end{cases}
$$

More recently, discriminators of various sequences were studied by
Schumer and Steinig \cite{Schumer&Steinig:1988}, Barcau \cite{Barcau:1988},
Schumer \cite{Schumer:1990},
Bremser, Schumer, and Washington \cite{Bremser&Schumer&Washington:1990},
Moree and Roskam \cite{Moree&Roskam:1995}
Moree \cite{Moree:1996},
Moree and Mullen \cite{Moree&Mullen:1996},
Zieve \cite{Zieve:1998},
Sun \cite{Sun:2013}, and
Moree and Zumalac\'arrequi \cite{Moree&Zumalacarregui:2016}.

In this paper we recall the definition of $k$-regular sequences, an
interesting class of sequences that has been widely studied.  We
introduce two well-known $2$-regular sequences, the so-called
``evil'' and ``odious'' numbers.  We prove that their discriminators
are $2$-regular.  Finally, we give an example of a $k$-regular
sequence whose discriminator sequence is not $k$-regular.

In the sequel, we use the following notation.
Let $\Sigma_k$ denote the alphabet $\{ 0, 1, \ldots, k-1 \}$.
If $x \in \Sigma_k^*$ is a string of digits,
then $[x]_k$ denotes the value of
$x$ when considered as a base-$k$ number.  If $n$ is an integer,
then $(n)_k$ is the string giving the canonical base-$k$ representation
of $n$ (with no leading zeroes).  If $x$ is a string of digits, then
$|x|$ denotes the length of the string $x$, and $|x|_a$ denotes
the number of occurrences of the letter $a$ in $x$.
Finally, $x^n = \overbrace{xx \cdots x}^n$ for $n \geq 0$.

By $S+i$, for $S$ a set of integers and $i$ an integer, we mean the
set $\{ x + i \ : \ x \in S\}$.
For sets $S$ and $T$, we write $S \sqcup T$ to denote the union
of $S$ and $T$, {\it and\/} the assertion
that this union is actually disjoint.

\section{$k$-regular sequences}

Let $k \geq 2$ be an integer.
The $k$-regular sequences are an interesting class of sequences with pleasant
closure properties \cite{Allouche&Shallit:1992,Allouche&Shallit:2003}.
They can be defined in several equivalent ways, and here we give three:

\medskip

-- They are the class of sequences 
$(s(n))_{n \geq 0}$ such that the set of 
subsequences of the form
$$ \{ (s(k^e n + i))_{n \geq 0} \ : \ 
e \geq 0 \text{ and } 0 \leq i < k^e \} $$
is a subset of a finitely-generated $\Zee$-module.

\medskip

-- They are the class of sequences $(s(n))_{n \geq 0}$
for which there exist an integer $r \geq 1$,
a $1 \times r$ row vector $u$,
an $r \times 1$ column vector $w$,
and an $r \times r$ matrix-valued morphism $\mu$ with domain $\Sigma_k^*$
such that $s(n) = u \mu(v) w$ for all strings $v$
with $[v]_k = n$.

\medskip

-- They are the class of sequences such that there are a finite
number of recurrence relations of the form
$$s(k^e n + i) = \sum_j a_j s(k^{e_j} n + i_j)$$
where $e \geq 0$, $e_j < e$, $0 \leq i < k^e$, and $0 \leq i_j < k^{e_j}$,
that completely determine all but finitely many values of $s$.

The $k$-regular sequences satisfy a number of nice closure properties.
We recall the following (see \cite{Allouche&Shallit:1992}):

\begin{theorem}
Let ${\bf r} = (r_i)_{i \geq 0}$ and ${\bf s} = (s_i)_{i \geq 0}$ be
two $k$-regular sequences of integers, and let $m \geq 1$ be an integer.
Then so are
\begin{itemize}
\item[(a)] ${\bf r} + {\bf s} = (r_i + s_i)_{i \geq 0}$;
\item[(b)] ${\bf r}{\bf s} = (r_i s_i)_{i \geq 0}$;
\item[(c)]  ${\bf r} \bmod m = (r_i \bmod m)_{i \geq 0}$.
\end{itemize}
\label{kreg}
\end{theorem}

\section{The evil and odious numbers}

The so-called ``evil'' and ``odious'' numbers are two examples of $2$-regular
sequences; they are sequences \seqnum{A001969} and \seqnum{A000069} in
Sloane's {\it On-Line Encyclopedia of Integer Sequences}, respectively.

The evil numbers $(\ev(n))_{n \geq 0}$ are 
$$ 0, 3, 5, 6, 9, 10, 12, 15, 17, 18, 20, 23, 24, 27, 29, 30, 33, 34, 36, 39, 40, 43, \ldots$$
and are those non-negative numbers
having an even number of $1$'s in their base-2 expansion.

The odious numbers $(\od(n))_{n \geq 0}$ are
$$ 1, 2, 4, 7, 8, 11, 13, 14, 16, 19, 21, 22, 25, 26, 28, 31, 32, 35, 37, 38, 41, \ldots$$
and are those non-negative numbers having
an odd number of $1$'s in their base-$2$ expansion.

Clearly the union of these two sequences is $\Enn$, the set of all
non-negative integers.

To see that these two sequences are
$2$-regular, note that both sequences satisfy
the recurrence relations
\begin{align*}
f(4n) &= -2f(n) + 3f(2n) \\
f(4n+1) &= -2f(n) + 2f(2n) + f(2n+1) \\
f(4n+2) &= {2\over 3} f(n) + {5\over 3} f(2n+1) \\
f(4n+3) &= 6 f(n) - 3f(2n) + 2f(2n+1) , 
\end{align*}
which can be proved by an induction using the characterization
in \cite[Example 12]{Allouche&Shallit:1992}.  

Let $\mathcal{O}_n = \{\od(i): \od(i) < n\}$ (resp.,
$\mathcal{E}_n = \{\ev(i): \ev(i) < n\}$)
denote the 
set of all odious (resp., evil) numbers that are strictly less than $n$.

\begin{lemma}
\begin{itemize}
\item[(a)] For $i \geq 1$ we have
$|\mathcal{O}_{2^{i}}| = |\mathcal{E}_{2^{i}}| = 2^{i-1}$.
\item[(b)] For $i \geq 1$ we have
$\mathcal{O}_{2^{i+1}} = \mathcal{O}_{2^{i}} \sqcup (\mathcal{E}_{2^{i}} + 2^i)$.
\item[(c)] For $i \geq 1$ we have
$\mathcal{E}_{2^{i+1}} = \mathcal{E}_{2^{i}} \sqcup (\mathcal{O}_{2^{i}} + 2^i)$.
\end{itemize}
\label{two}
\end{lemma}

\begin{proof}
\begin{itemize}
\item[(a)]
Let $0 \leq n < 2^i$.  These $n$ can be placed in 1--1 
correspondence with the binary strings $w$ of length $i$, using
the correspondence $[w]_2 = n$.  For each binary string $x$ of
length $i-1$, either $x0$ is odious and $x1$ is evil, or
vice versa.  Thus there are $2^{i-1}$ odious numbers less than
$2^i$, and $2^{i-1}$ evil numbers less than $2^i$.
\item[(b)]  Let $2^i \leq n < 2^{i+1}$.   
Consider $n - 2^i$.  Since the base-$2$
expansion of $n-2^i$ differs from that of $n$ by omitting the first
bit, clearly $n - 2^i$ is evil iff $n$ is odious.  
\item[(c)]  Just like (b).
\end{itemize}
\end{proof}

This gives the following corollary:

\begin{corollary}
For integers $n \geq 0$ and $i \geq 1$
we have $\od (n) \in \mathcal{O}_{2^i}$
and $\ev (n) \in \mathcal{E}_{2^i}$ if and only if $n < 2^{i-1}$. 
Furthermore
\begin{align}
\od (2^{i-1}) &= 2^i;\\
\ev (2^{i-1}) &= 2^i + 1.
\end{align}
\end{corollary}

\subsection{Discriminator of the odious numbers}

We now turn our attention to the discriminators for the evil
and odious numbers, starting with the odious numbers.
First, we need to prove the following useful lemma.

\begin{lemma}
Let $i \geq 1$ and $1 \leq m < 2^i$.  Then there exist two odious numbers
$j, \ell$ with $1 \leq j < \ell \leq 2^i$ such that
$m = \ell - j$.
\label{three}
\end{lemma}

\begin{proof}
Let $w = (m)_2$.  There are three cases according to the form of $w$.

\begin{enumerate}

\item No $1$ follows a $0$ in $w$.  Then $w = 1^{a}0^{b}$,
where $a \ge 1$, $b \ge 0$, and $a+b \leq i$.  So $m = 2^b(2^a - 1)$.
Take $\ell = 2^{a+b}$ and $j = 2^b$.

\item $w = x01y$, where $|xy|_1$ is odd.
Take $j = 2^{|y|+1}$.
Then $\ell = m + 2^{|y|+1}$. 
Now $(\ell)_2 = x11y$, and clearly $|x11y|_1$ is odd, so $\ell$ is odious.

\item $w = x01y$, where $|xy|_1$ is even.
Take $j = 2^{|y|}$.  Then
$\ell = m + 2^{|y|}$.
Now $(\ell)_2 = x10y$, and clearly $|x10y|_1$ is odd, so $\ell$ is odious.
\end{enumerate}
\end{proof}

With the help of this lemma, we can compute the discriminator for the sequence of odious numbers.
\begin{theorem}
For the sequence of odious numbers,
the discriminator $D_{\od} (n)$ satisfies the equation
\begin{equation}
D_{\od} (n) = 2^{\lceil \log_2 n \rceil}
\end{equation}
for $n \ge 1$.
\end{theorem}

\begin{proof}
The cases $n = 1, 2$ are left to the reader.
Otherwise, let $i \geq 1$ be such that
$2^i < n \leq 2^{i+1}$.  We show $D_{\od} (n) = 2^{i+1}$.
There are two cases:

\medskip

Case 1:  $n = 2^i + 1$.   We must compute the discriminator of
$\od(0), \od(1), \ldots, \od(2^i) = 2^{i+1}$.  
By Lemma~\ref{three}, for each $m < 2^{i+1}$, there exist
odious numbers $j, \ell$ with $1 \leq j \leq l \leq 2^{i+1}$
with $\ell - j = m$.  So the numbers
$\od(0), \od(1), \ldots, \od(2^i)$ cannot be pairwise incongruent
mod $m$ for $m < 2^{i+1}$.  On the other hand, since
$0$ is not odious and each of the numbers
$\od(0), \od(1), \ldots, \od(2^i)$ are less than $2^{i+1}$ except the very
last (which is $0$ mod $2^{i+1})$, clearly $2^{i+1}$ 
discriminates $\od(0), \od(1), \ldots, \od(2^i)$.  

\medskip

Case 2:  $2^i + 1 <  n \leq 2^{i+1}$.  Since the discriminator is
nondecreasing, we know $D_{\od} (n) \geq 2^{i+1}$.  It suffices to show
that $2^{i+1}$ discriminates
$\mathcal{O}_{2^{i+2}} = \{ \od(0), \od(1), \ldots, \od(2^{i+1} - 1) \}$.    
Now from Lemma~\ref{two}(b), we have 
$$\mathcal{O}_{2^{i+2}} = 
\mathcal{O}_{2^{i+1}} \sqcup (\mathcal{E}_{2^{i+1}} + 2^{i+1}).$$
If we now take both sides modulo $2^{i+1}$, we see that the
right-hand side is just 
$\mathcal{O}_{2^{i+1}} \sqcup \mathcal{E}_{2^{i+1}}$,
which represents all integers in the range $[0, 2^{i+1})$.

\end{proof}

Empirically, many interesting sequences of positive integers seem to have
discriminator $2^{\lceil \log_2 n \rceil}$.  However, of all such
sequences, the odious numbers play a special role:  they are the
lexicographically least.

\begin{theorem}
The sequence of odious numbers is the lexicographically least increasing
sequence of positive integers $\bf s$ such that
$D_{s} (n) = 2^{\lceil \log_2 n \rceil}$.
\end{theorem}

\begin{proof}
We prove this by contradiction. Suppose there exists a sequence of
increasing positive integers, $s(0), s(1), \ldots$, that is
lexicographically smaller than the sequence of odious numbers but
shares the same discriminator, $D_{s} (n) = 2^{\lceil \log_2 n
\rceil}$.

Let $j$ denote the first index such that $s(j) \neq \od(j)$, i.e.,
$s(j) < \od(j)$, since $s$ is a lexicographically smaller sequence than
the odious numbers. We can see that $s(j)$ must be evil, because
$\od(j)$ is the next odious number after $\od(j - 1) = s(j - 1)$. Note
that since $\od(0) = 1$ is the smallest positive integer, necessarily
$j \geq 1$.

Now let $i \geq 0$ be such that $2^i \leq j < 2^{i + 1}$.
In that case, the discriminator of the sequence $s(0), s(1), \ldots,
s(j)$ is $D_{s} (j + 1) = 2^{\lceil \log_2 (j + 1) \rceil} = 2^{i
+ 1}$. However, $s(j)$ also discriminates this sequence, which
implies that $s(j) \geq D_s (j + 1) = 2^{i + 1}$. Note that by
the definition of $j$, this means that all odious numbers less than $2^{i +
1}$ are present in the sequence $s(0), s(1), \ldots, s(j)$.

Furthermore, we have $s(j) < \od(j) < \od(2^{i + 1}) = 2^{i + 2}$.
So $2^{i + 1} \leq s(j) <
2^{i + 2}$, which means that the largest power of $2$ appearing in the
binary representation of $s(j)$ is $2^{i+1}$.
Therefore $s(j) \bmod 2^{i +
1} = s(j) - 2^{i + 1}$ is odious.
However, $s(j) \bmod 2^{i + 1} < 2^{i + 1}$.
But the sequence $s(0), s(1),
\ldots, s(j)$ contains all odious numbers less than $2^{i + 1}$, which
therefore includes the result of $s(j) \bmod 2^{i + 1}$. In other
words, $s(j)$ is congruent to another number in this sequence modulo
$2^{i + 1}$, i.e., $D_s (j + 1) \neq 2^{i + 1}$, which is a
contradiction.
\end{proof}

\subsection{Discriminator of the evil numbers}

We now focus on the discriminator for the sequence of evil numbers. Here, we need to utilize a similar lemma as before.

\begin{lemma}
Let $i \geq 3$ and $1 \leq m < 2^i - 3$.  Then there exist two evil numbers
$j, \ell$ with $0 \leq j < \ell \leq 2^i + 1$ such that
$m = \ell - j$.
\label{four}
\end{lemma}

\begin{proof}
Let $w = (m)_2$.  There are several cases according to the form of $w$.

\begin{enumerate}

\item The number $m$ is evil. Take $\ell = m$ and $j = 0$.

\item There are no $0$'s in $w$. Then $m = 2^a - 1$ where $0 < a < i$. Note that $a \neq i$. If $m = 1$, then take $\ell = 6$ and $j = 5$. Otherwise, take $\ell = 2^a + 2$ and $j = 3$.

\item No $1$ follows a $0$ in $w$ and $|w|_0 > 0$.  Then $w = 1^{a}0^{b}$, where $a \ge 1$, $b \ge 1$, and $a + b \leq i$.  So $m = 2^b(2^a - 1)$. Take $\ell = 2^{a+b} + 1$ and $j = 2^b + 1$.

\item There is exactly one 0 in $w$ and $w$ ends with 01. Then $w = 1^{a}01$, where $1 \le a \le i - 3$. So $m = 2^{a + 2} - 3$. Take $\ell = 2^{a + 2} + 2$ and $j = 5$.

\item There is exactly one 0 in $w$ and $w$ ends with 11. Then $w = 1^{a}01^{b}$, where $a \ge 1$, $b \ge 2$, and $a + b \leq i - 1$. So $m = 2^{a + b + 1} - 2^b - 1$. Take $\ell = 2^{a + b + 1} + 1$ and $j = 2^b + 2$.

\item $w = x01y0z$, where $|xyz|_1$ is even. Take $j = 2^{|y| + |z| + 1} + 2^{|z|}$. Then $\ell = m + 2^{|y| + |z| + 1} + 2^{|z|}$. Now $(\ell)_2 = x10y1z$. We can see $|x10y1z|_1$ is even, so $\ell$ is evil.

\item $w = x0y01z$, where $|xyz|_1$ is even. Take $j = 2^{|y| + |z| + 2} + 2^{|z|}$. Then $\ell = m + 2^{|y| + |z| + 2} + 2^{|z|}$. Now $(\ell)_2 = x1y10z$. We can see $|x1y10z|_1$ is even, so $\ell$ is evil.
\end{enumerate}
\end{proof}

With the help of this lemma, we can compute the discriminator for the sequence of evil numbers.
\begin{theorem}
For the sequence of evil numbers, the discriminator $D_{\ev} (n)$ 
satisfies the equation
\begin{equation}
D_{\ev} (n) = \begin{cases}
2^{i + 1} - 3, & \text{if } n = 2^i + 1 \text{ for odd } i \geq 2;\\
2^{i + 1} - 1, & \text{if } n = 2^i + 1 \text{ for even } i \geq 2;\\
2^{\lceil \log_2 n \rceil}, & \text{otherwise,} 
\end{cases}
\end{equation}
for $n \ge 1$.
\end{theorem}

\begin{proof}
The cases $n = 1, 2, 3, 4$ are left to the reader.
Otherwise, let $i \geq 2$ be such that
$2^i < n \leq 2^{i+1}$.  We show $D_{\ev} (n)$ satisfies the given equation.
There are three cases presented in the equation:

\medskip

Case 1:  $n = 2^i + 1$ for odd $i \geq 2$.   We must compute the discriminator of $\ev(0), \ev(1), \ldots,$ $\ev(2^i) = 2^{i+1} + 1$.  
By Lemma~\ref{four}, for each $m < 2^{i+1} - 3$, there exist
evil numbers $j, \ell$ with $1 \leq j \leq l \leq 2^{i+1} + 1$
with $\ell - j = m$.  So the numbers
$\ev(0), \ev(1), \ldots, \ev(2^i)$ cannot be pairwise incongruent
mod $m$ for $m < 2^{i + 1} - 3$. 
Note that for odd $i \geq 2$, the only evil numbers in the range $[2^{i + 1} - 3, 2^{i + 1} + 1]$ are $2^{i + 1} - 1$ and $2^{i + 1} + 1$, easily observed from their binary representations. We can see that $2^{i + 1} - 1 
\equiv \modd {2} {2^{i + 1} - 3}$ and
$2^{i + 1} + 1 \equiv \modd{4} {2^{i + 1} - 3}$,
where neither $2$ nor $4$ are evil.
All the other numbers in the sequence $\ev(0), \ev(1), \ldots, \ev(2^i)$ are less than $2^{i + 1} - 3$, and thus it is clear that $2^{i + 1} - 3$ discriminates $\ev(0), \ev(1), \ldots, \ev(2^i)$.

\medskip

Case 2:  $n = 2^i + 1$ for even $i \geq 2$.   We must compute the discriminator of $\ev(0), \ev(1), \ldots,$ $\ev(2^i) = 2^{i+1} + 1$.  
Just as in the previous case, Lemma~\ref{four} ensures that the numbers
$\ev(0), \ev(1), \ldots, \ev(2^i)$ cannot be pairwise incongruent
mod $m$ for $m < 2^{i + 1} - 3$. For even $i \geq 2$, we can see that both $2^{i + 1} - 3$ and $2^{i + 1} - 2$ are evil from their binary representations. Neither of them can discriminate the sequence since $m \bmod m = 0$ for either $m = 2^{i + 1} - 3$ or $m = 2^{i + 1} - 2$, while $0$ is evil. Thus the discriminator must be at least $2^{i + 1} - 1$. Since neither $2^{i + 1} - 1$ nor $2^{i + 1}$ are evil, we can see that each of the numbers $\ev(0), \ev(1), \ldots, \ev(2^i)$ are all less than $2^{i + 1} - 1$ except the very last, which is $2^{i + 1} + 1 = 2 \mod (2^{i + 1} - 1)$, where $2$ is not evil. Therefore, it is clear that $2^{i + 1} - 1$ discriminates $\ev(0), \ev(1), \ldots, \ev(2^i)$.

\medskip

Case 3:  $2^i + 1 <  n \leq 2^{i+1}$. From the previous two cases, we know that $D_{\ev} (2^i + 1)$ is either $2^{i + 1} - 3$ or $2^{i + 1} - 1$. Since the discriminator is nondecreasing, we know $D_{\ev} (n) \geq 2^{i+1} - 3$. We can see that the sequence $\ev(0), \ev(1), \ldots, \ev(n - 1)$ must include $\ev(2^i + 2) = 2^{i + 1} + 2$, the next evil number after $2^{i + 1} + 1$. We then observe that
\begin{align*}
&2^{i + 1} + 2 \equiv \modd{5} {2^{i + 1} - 3},\\
&2^{i + 1} + 1 \equiv \modd{3} {2^{i + 1} - 2},\\
&2^{i + 1} + 2 \equiv \modd{3} {2^{i + 1} - 1},
\end{align*}
where the numbers $3$ and $5$ are evil. Therefore, the discriminator must be at least $2^{i + 1}$. It suffices to show that $2^{i + 1}$ discriminates
$\mathcal{E}_{2^{i+2}} = \{ \ev(0), \ev(1), \ldots, \ev(2^{i + 1} - 1) \}$.    
Now from Lemma~\ref{two}(c), we have 
$$\mathcal{E}_{2^{i+2}} = 
\mathcal{E}_{2^{i+1}} \sqcup (\mathcal{O}_{2^{i+1}} + 2^{i+1}).$$
If we now take both sides modulo $2^{i+1}$, we see that the
right-hand side is just 
$\mathcal{E}_{2^{i+1}} \sqcup \mathcal{O}_{2^{i+1}}$,
which represents all integers in the range $[0, 2^{i+1})$. Thus we have $D_{\ev} (n) = 2^{i + 1} = 2^{\lceil \log_2 n \rceil}$ for $2^i + 1 <  n \leq 2^{i+1}$.

\end{proof}

\section{A $k$-regular sequence whose discriminator is not $k$-regular}

Consider the sequence $1, 4, 9, 16, \ldots$ of perfect squares.
From \cite[Example 5]{Allouche&Shallit:1992}, 
this sequence is $k$-regular for all integers $k \geq 2$.  We show

\begin{theorem}
The discriminator sequence of the perfect squares is not $k$-regular
for any $k$.
\end{theorem}

\begin{proof}
We use the characterization of the discriminator sequence
$D_{\rm sq}(n)$ given above in Section~\ref{intro}.
Suppose $D_{\rm sq}(n)$ is $k$-regular.  Then
from Theorem~\ref{kreg} (c) we know that the sequence $A$ given by
$A(n) = D_{\rm sq}(n) \bmod 2$ is $k$-regular.    From
Theorem~\ref{kreg} (b)  we know that the sequence
$F(n) = A(n) D_{\rm sq}(n)$ is $k$-regular.  From Theorem~\ref{kreg} (a)
we know that the sequence $B(n) = 2 - 2A(n)$ is $k$-regular.
From Theorem~\ref{kreg} (a) we know that the sequence
$E(n) = F(n)+B(n)$ is $k$-regular.  It is now easy to see that
for $n > 4$ we have $E(n) = 2$ if $B(n)$ is even, while
$E(n) = D_{\rm sq}(n)$ if $D_{\rm sq}(n)$ is odd.    Thus $E(n)$ takes only
prime values for $n > 4$.

We now argue that $E(n)$ is unbounded.  To see this, it suffices
to show that there are infinitely many indices $n$
such that $D_{\rm sq}(n)$ is prime.   By Dirichlet's theorem on primes
in arithmetic progressions there are infinitely many primes $p$
for which $p \equiv 1$ (mod $4$).  For these primes consider
$n = (p-1)/2$.  Then $2n = p-1$ is divisible by $4$ and hence
not twice a prime, but $2n+1 = p$.  Hence
for these $n$ we have $D_{\rm sq}(n) = p = 2n+1$,
and hence $E(n) = D_{\rm sq}(n)$.
Thus $(E(n))$ is unbounded.

Finally, we apply a theorem of Bell \cite{Bell:2005} to the sequence $E$.  
Bell's theorem states that any unbounded $k$-regular sequence must take
infinitely many composite values.  However, the sequence $(E(n))$ is
unbounded and takes only prime values for $n > 4$.  This contradiction
shows that $D_{\rm sq}(n)$ cannot be $k$-regular.
\end{proof}

\section{Discriminator of the Cantor numbers}

Consider the Cantor numbers $(C(n))_{n \geq 0}$ 
$$ 0,2,6,8,18,20,24,26,54,56,60,62,72,74,78,80,162,164,168,170,180, \ldots $$
which are the numbers having only $0$'s and $2$'s in their base-$3$
expansion.   This is sequence \seqnum{A005823} in Sloane's
{\it On-Line Encyclopedia of Integer Sequences}.    It is
$2$-regular, as it satisfies the recurrence relations
\begin{align*}
C(2n) &= 3C(n) \\
C(2n+1) &= 3C(n) + 2.
\end{align*}

We have the following conjecture about the discriminator sequence
$D_C (n)$ of the Cantor numbers:

\begin{align*}
D_C(8n) &=  {{13} \over 3} D_C(4n) - 2 D_C(4n+1)  + {2 \over 3} D_C(4n+2) \\
D_C(8n+1) &= {3 \over 2} D_C(2n) + {7 \over 2} D_C(4n) -
	2 D_C(4n+1) + D_C(4n+2) \\
D_C(8n+2) &= {{10} \over 3} D_C(4n) - 2 D_C(4n+1)  + {5 \over 3} D_C(4n+2) \\
D_C(8n+3) &= {9 \over 2} D_C(2n) + {{11} \over 6} D_C(4n) -
        3 D_C(4n+1) + {8 \over 3} D_C(4n+2) \\
D_C(8n+4) &= 6 D_C(2n) -2 D_C(4n) + 2 D_C(4n+1) + D_C(4n+2) \\
D_C(8n+5) &= 6 D_C(2n) -2 D_C(4n) + D_C(4n+1) + 2 D_C(4n+2) \\
D_C(8n+6) &= {3 \over 2} D_C(2n) - {1 \over 2} D_C(4n) - D_C(4n+1) + 4 D_C(4n+2) \\
D_C(16n+7) &= -3 D_C(2n) + D_C(4n) + 7 D_C(4n+1) + 2D_C(4n+2)  \\
D_C(16n+15) &= -9 D_C(n) + {{27} \over 2} D_C(2n) - {{15} \over 2} D_C(4n)
	+9 D_C(4n+1) - 6 D_C(4n+2) + 10 D_C(4n+3). \\
\end{align*}
If true, this would mean that $D_C(n)$ is also $2$-regular.

\section{Self-discriminators}

In this section we change our indexing slightly.
Let ${\bf s} = (s_1, s_2, s_3, \ldots)$ be an increasing
sequence of positive integers and let
$D_{\bf s} = (d_1, d_2, d_3, \ldots )$ be its
associated discriminator sequence.
When does ${\bf s} = D_{\bf s}$?

\begin{theorem}
The sequence $\bf s$ is its own discriminator if and only if
either $s_i = i$ for all $i \geq 1$, or the following three
conditions hold:
\begin{itemize}
\item[(a)] There exists an integer $t \geq 1$ such that
$s$ begins $1,2,3, \ldots, t$ but not $1,2,3\ldots, t+1$; and
\item[(b)] $s_{t+1} \in \{ t+2, \ldots, 2t+1 \}$; and
\item[(c)] $s_{i+1} - s_i \in \{1,2, \ldots, t\}$ for $i > t$.
\end{itemize}
\end{theorem}

\begin{proof}
The case where $s_i = i$ for all $i$ is easy and is left to the reader.
Similarly, if $s_1 > 1$ then $\bf s$ cannot be its own discriminator.
Otherwise, assume condition (a) holds.

Let $(d_i)_{i \geq 1}$ be the discriminator of ${\bf s} = (s_i)_{i \geq 1}$.
We show that conditions (b) and (c) hold iff
$d_i = s_i$ for all $i$.  If $1 \leq i \leq t$, this is clear.
There are two cases to consider.

\medskip

$i = t+1$:  assume $d_{t+1} = s_{t+1}$.  
Since $(s_i)_{i \geq 1}$ is increasing, from (a) we have
$s_{t+1} \geq t+2$ and so $d_{t+1} \geq t+2$.  On the other hand, if 
$s_{t+1} \geq 2t+2$, then the sequence $(1,2,\ldots, t, s_{t+1})$ is
discriminated by $s_{t+1} - (t + 1)$, a contradiction.  So
$s_{t+1} \leq 2t+1$.

For the other direction,
suppose $s_{t+1} \in \{ t+2, \ldots, 2t+1 \}$.  Then
contrary to what we want to prove, if
$t \leq d_{t+1} < s_{t+1}$, then $s_{t+1} \equiv
\modd{s_{t+1} - d_{t+1}} {d_{t+1}}$.  
Unless $s_{t+1} = 2t+1$, $d_{t+1} = t$, we have
$1 \leq s_{t+1} - d_{t+1} \leq t$, a contradiction,
since then $s_{t+1} \bmod d_{t+1}$ already occurred
in $s_1 \bmod d_{t+1}, s_2 \bmod d_{t+1}, \ldots, s_t \bmod d_{t+1}$.
In the exceptional case $s_{t+1} = 2t+1$, $d_{t+1} = t$, which gives
$s_{t+1} \bmod d_{t+1} = 1 = s_1$, a contradiction.
So $d_{t+1} \geq s_{t+1}$.
On the other hand,
it is easy to see that $s_{t+1}$ discriminates $s_1, s_2, \ldots, s_{t+1}$.  

\medskip

$i > t+1$:
Suppose $d_i = s_i$ for all $i > t+1$, and, to get a contradiction,
let $i >  t+1$ be the smallest index for which 
$s_i - s_{i-1} \not\in \{1,2, \ldots, t\}$.  
We cannot have $s_i = s_{i-1}$ because the sequence $(s_i)_{i \geq1}$
is strictly increasing.  So $d_i = s_i \geq s_{i-1} + t + 1$.  
But then the sequence $(s_1, s_2, \ldots, s_i)$
is also discriminated by $s_i - (t+1)$, a contradiction.
So $s_i - s_{i-1} \in \{1,2, \ldots, t\}$, as claimed.

For the other direction, assume
$s_{i} - s_{i-1} \in \{1,2, \ldots, t\}$ and $d_i < s_i$.  Then 
${s_i \bmod d_i  } = s_i - d_i < t$, a contradiction.
On the other hand, $s_i \bmod d_i = 0$, which is not an element of
$\bf s$, so $d_i$ discriminates $s_1, \ldots, s_i$, as desired.
\end{proof}

\begin{corollary}
There are uncountably many increasing sequences of positive integers
that are their own discriminators.
\end{corollary}

\begin{corollary}
For $1 \leq t \leq n$,
the number of length-$n$ finite sequences,
beginning with $1, 2, \ldots, t$ but not
$1, 2, \ldots, t, t+1$, that are self-discriminators
is $t^{n-t}$.  Hence the total number of finite sequences of length
$n$ that are self-discriminators is
$\sum_{1 \leq t \leq n} t^{n-t}$.  
\end{corollary}

\begin{proof}
Suppose $s_1, s_2, \ldots, s_t$ are fixed.  If $t = n$, there is
exactly one such sequence.   Otherwise $s_{t+1}$ is
constrained to lie in $\{ t+2, \ldots, 2t+1 \}$, which is of
cardinality $t$, and subsequent terms $s_i$ (if there are any)
are constrained to lie
in $\{ s_{i-1} + 1, \ldots, s_{i-1} + t \}$, which is also of
cardinality $t$.  There are $n-t$ remaining
terms, which gives $t^{n-t}$ possible extensions of length $n$.
\end{proof} 
\begin{remark}
The number of finite sequences of length $n$ that are self-discriminators
is given by sequence \seqnum{A026898} in Sloane's 
{\it On-Line Encyclopedia of Integer Sequences} \cite{Sloane:2016}.
\end{remark}

\section{Acknowledgment}

We are grateful to Pieter Moree for having introduced us to this 
interesting topic.  The idea of considering the discriminator for $k$-regular related sequences is due to Maike Neuwohner.
She did this during an internship with Pieter Moree, who challenged her
to find sequences having an a discriminator 
displaying a well-describable behaviour. In particular she determined
the discriminator for a class of sequences related to 
Szekeres' sequence \seqnum{A003278}; see \cite{Neuwohner:2015}.

\end{document}